\documentclass[11pt,a4paper]{article}

\usepackage{amsmath,amssymb,amsfonts,amsthm,,a4}
\usepackage{url}
\usepackage{multirow}
\usepackage{tikz}

\newtheorem{conjecture}{Conjecture}
\newtheorem{theorem}{Theorem}
\newtheorem{lemma}{Lemma}
\newtheorem{proposition}{Proposition}

\newtheorem{remark}{Remark}
\newtheorem{example}{Example}

\usepackage{float}

\usepackage{hyperref}
\hypersetup{
    unicode=false,          
    pdftoolbar=true,        
    pdfmenubar=true,        
    pdffitwindow=true,      
    pdftitle={A $\tau$-conjecture for Newton polygons},    
    pdfauthor={},     
    pdfsubject={},   
    pdfcreator={Pascal Koiran},   
    pdfproducer={}, 
    pdfkeywords={compressed sensing, restricted isometry property, basis pursuit
}, 
    pdfnewwindow=true,      
    colorlinks=false,       
    linkcolor=red,          
    citecolor=green,        
    filecolor=magenta,      
    urlcolor=cyan          
}

\usepackage{algorithm}
\usepackage{algpseudocode}

\newcommand\zz{\ensuremath{\mathbb{Z}}}
\newcommand\kk{\ensuremath{\mathbb{K}}}

\newcommand\p{\ensuremath{\mathsf P}}
\newcommand\vp{\ensuremath{\mathsf{VP}}}
\newcommand\np{\ensuremath{\mathsf{NP}}}
\newcommand\vnp{\ensuremath{\mathsf{VNP}}}

\newcommand\Newt{\ensuremath{\mathrm{Newt}}}
\newcommand\Mon{\ensuremath{\mathrm{Mon}}}
\newcommand\conv{\ensuremath{\mathrm{conv}}}

\title{A $\tau$-Conjecture for Newton Polygons\\
}

\author{Pascal Koiran, Natacha Portier, S\'ebastien Tavenas, St\'ephan Thomass\'e\\
LIP\thanks{UMR 5668 ENS Lyon, CNRS, UCBL, INRIA.
Email: {\tt [Pascal.Koiran, Natacha.Portier, Sebastien.Tavenas, Stephan.Thomasse]@ens-lyon.fr.} The authors are supported by ANR project CompA.}, \'Ecole Normale Sup\'erieure de Lyon, Universit\'e de Lyon.
}

\begin{document}

\maketitle

\begin{abstract}
One can associate to any bivariate polynomial $P(X,Y)$ 
its Newton polygon. This is 
the convex hull of the points $(i,j)$ such that the monomial $X^i Y^j$
appears in $P$ with a nonzero coefficient.
We conjecture that when~$P$ is expressed as a sum of products of 
sparse polynomials, the number of edges of its Newton polygon 
is polynomially bounded in the size of such an expression. 
We show that this ``$\tau$-conjecture for Newton polygons,'' even in a weak form,  implies
that the permanent polynomial is not computable by polynomial 
size arithmetic circuits. We make the same observation for a weak version 
of an earlier ``real $\tau$-conjecture.''
Finally, we make some progress 
toward the $\tau$-conjecture for  
Newton polygons using recent results
from combinatorial geometry.
\end{abstract}

\section{Introduction}

Let $f \in \zz[X]$ be a univariate polynomial computed by an arithmetic circuit
(or equivalently, a straight-line program) of size $s$ starting from the
variable $X$ and the constant $1$.
According to Shub and Smale's $\tau$-conjecture~\cite{ShubSmale}, 
the number of integer roots
of $f$ should be bounded by a fixed polynomial function of $s$.
It was shown in~\cite{ShubSmale} that the $\tau$-conjecture implies
a $\p \neq \np$ result for the Blum-Shub-Smale model of computation 
over the complex numbers~\cite{BSS89,BCSS96}.
A similar result was obtained by B\"urgisser~\cite{Burg09} 
for another algebraic version
of $\p$ versus $\np$ put forward by Valiant~\cite{Valiant79,Valiant82} 
at the end of the 1970's.
A succinct way of stating this $\vp$ versus $\vnp$ problem goes as follows: 
can we compute the permanent of a $n \times n$ matrix  with a  number  
of arithmetic operations which is polynomial in~$n$? 
This question can be formalized using the computation model of arithmetic
circuits. The permanent plays a special role here 
because it is $\vnp$-complete, and it can be replaced by any other $\vnp$-complete family of polynomials.
We refer to B\"urgisser's book~\cite{Burg} 
for an introduction to this topic and
to two recent surveys on arithmetic circuit complexity 
by Shpilka and Yehudayoff~\cite{SY10} 
and by Chen, Kayal and Wigderson~\cite{CKW11}.

As a natural approach to the $\tau$-conjecture, one can try to bound 
the number of real roots instead of the number of integer roots.
This fails miserably since the number of real roots of a univariate polynomial
can be exponential in its arithmetic circuit size. Chebyshev polynomials
provide such an example~\cite{SmaleProblems} (a similar example was provided
earlier by Borodin and Cook~\cite{BC76}).
A real version of the $\tau$-conjecture was nevertheless proposed 
in~\cite{Koi10a}. In order to avoid the aforementioned counterexamples, 
the attention is restricted to arithmetic circuits of
a special form: the sums of products of sparse polynomials.
In spite of this restriction, the real $\tau$-conjecture still implies
that the permanent is hard to compute 
for general arithmetic circuits~\cite{Koi10a}.

In this paper, we propose a $\tau$-conjecture for Newton polygons of bivariate polynomials.
Like the real $\tau$-conjecture, it deals with sums of products 
of sparse polynomials and implies that the permanent is hard 
for general arithmetic circuits.
A common idea to these three $\tau$-conjecture is that ``simple'' arithmetic 
circuits should compute only ``simple'' polynomials. 
In the original $\tau$-conjecture, the simplicity of a polynomial
is measured by the number of its integer roots; in the real $\tau$-conjecture
it is measured by the number of its real roots; and in our new conjecture 
by the number of edges of its Newton polygon. These conjectures
are independent in the sense that we know of no implication
between two of them.

\subsection*{Organization of the paper}

In Section~\ref{polygons} we review some basic facts about Newton polygons and
formulate the corresponding $\tau$-conjecture.
We also state in Theorem~\ref{permanent} the motivating result for this paper: a proof of the conjecture,
even in a very weak form, implies a lower bound for the permanent.
In Section~\ref{proof} we give a proof of this result and of a refinement:
it suffices to work with sums of {\em powers} of sparse polynomials
rather than with sums of arbitrary products. 
We also point out that this refinement applies to the real $\tau$-conjecture
from~\cite{Koi10a}, and that (like in Theorem~\ref{permanent}) 
a very weak form of this conjecture implies a lower bound for the permanent.
These observations improve the results stated in~\cite{Koi10a}.
In Section~\ref{bounds} we use a recent result 
of convex geometry~\cite{EPRS08}
to provide nontrivial 
upper bounds on the number of edges
of Newton polygons. Our results fall short of establising the new
$\tau$-conjecture  (even in the weak form required
by Theorem~\ref{permanent}) but they improve significantly 
on the naive bound obtained by brute-force expansion.
For instance, as a very special case of our results we have that
the Newton polygon of $fg+1$ has $O(t^{4/3})$ edges if the bivariate 
polynomials $f$ and $g$ have at most $t$ monomials. 
The straightforward bound obtained by expanding the 
product $fg$ is only $O(t^2)$.
We conclude the paper with a couple of open problems.
In particular, we ask whether this $O(t^{4/3})$ upper bound is optimal.
In the appendix, we improve on this upper bound by giving a linear upper bound
in a  special case.

\section{Newton Polygons} \label{polygons}

In the following we will consider a field $\kk$ of characteristic
$0$. In fact, most of the results of this paper are valid for a general field.

We first recall some standard background on Newton polygons. Much more can be found in the survey~\cite{Sturmfels98}.
Consider a bivariate polynomial  $f  \in \kk[X,Y]$. To each monomial $X^i Y^j$ appearing in $f$ with a nonzero coefficient we associate the point with coordinates $(i,j)$ in the Euclidean plane. 
We denote by $\Mon(f)$ this finite set of points.
If $A$ is a set of points in the plane, we denote by $\conv(A)$ the convex
hull of $A$.
By definition the Newton polygon of $f$, denoted $\Newt(f)$, is the convex hull of 
$\Mon(f)$ (In particular, $\Newt(f)=\conv(\Mon(f))$).
Note that $\Newt(f)$ has at most $t$ edges if $f$ has $t$ monomials.
 It is well known~\cite{Ostro75} that the Newton polygon of a product of polynomials is the Minkowski sum of their Newton polygons, i.e.,
$$\Newt(fg)=\Newt(f)+\Newt(g)=\{p+q;\ p \in \Newt(f),\ q \in \Newt(g)\}.$$
A short proof of this fact can be found in~\cite{Gao01}.
As a result, if $f$ has $s$ monomials and $g$ has $t$ monomials 
then $\Newt(fg)$ has at most $s+t$ edges.
More generally, for a product $f=g_1g_2\cdots g_m$,
$\Newt(f)$ has at most $\sum_{i=1}^m t_i$ edges where $t_i$ is the number
of monomials of $g_i$; but $f$ can of course have up to $\prod_{i=1}^m t_i$
monomials.
The number of edges of a Newton polygon is therefore easy
to control for a product of polynomials.
In contrast,  the situation is not at all clear
for a sum of products.
We propose the following conjecture.

\begin{conjecture}[$\tau$-conjecture for Newton polygons] \label{tauconj}
There is a polynomial $p$ such that the following property holds.

Consider   any bivariate polynomial
 $f  \in \kk[X,Y]$ 
 of the form
\begin{equation} \label{sps}
f(X,Y)=\sum_{i=1}^k \prod_{j=1}^m f_{ij}(X,Y)
\end{equation}
where the $f_{ij}$ have at most $t$ monomials.
Then the Newton polygon  of  $f$ has at most $p(kmt)$ edges.
\end{conjecture}
Note that the bound only depends on $k$, $m$ and $t$, and so does
not depend on the degrees and coefficients of the polynomials $f_{i,j}$.

The ``real $\tau$-conjecture''~\cite{Koi10a} is a similar conjecture
for real roots of sums of products of sparse\footnote{Here and
  in~\cite{Koi10a}, the term ``sparse'' refers to the fact that we
  measure the size of a polynomial $f_{ij}$ by the number of its
  monomials.} univariate polynomials, and it implies that the
permanent does not have polynomial-size arithmetic circuits. We are not
able to compare these two conjectures. However, as we shall see shortly,
Conjecture~\ref{tauconj} also implies that the permanent does not have
polynomial-size arithmetic circuits.

One can also formulate a $\tau$-conjecture for the multiplicities of
(complex) nonzero roots of univariate polynomials
(see Corollary~2.4.1 in~\cite{Hrubes13}). This
conjecture again implies that the permanent is hard for arithmetic circuits.

By expanding the products in~(\ref{sps}) we see that $f$ has at most $k.t^m$ monomials,
and this is an upper bound on the number of edges of its Newton polygon.
In  order to improve this very coarse bound, the main difficulty 
we have to face is that the $k$-fold summation in the definition of 
$f$ may lead to cancellations of monomials. As a result, some of the vertices
of $\Newt(f)$ might not be vertices of the Newton polygons 
of any of the $k$ products $\prod_{j=1}^m f_{ij}(X,Y)$.
We give two examples of such cancellations below.
If there are no cancellations (for instance, if the $f_{ij}$ only have positive coefficients) then we indeed have a polynomial upper bound. 
In this case, $\Newt(f)$ is the convex hull of the union
of the Newton polygons of the $k$ products. Each of these $k$ Newton polygons
has at most $mt$ vertices, so $\Newt(f)$ has at most $kmt$ vertices and
as many edges.

\begin{example}
We define $A(X,Y)=XY+X^2+X^2Y^2+X^3Y+X^5Y$, $B(X,Y)=1+XY^2$, $C(X,Y)=-X-XY-X^2Y^2$
and $D(X,Y)=Y+X+X^2Y+X^4Y$.
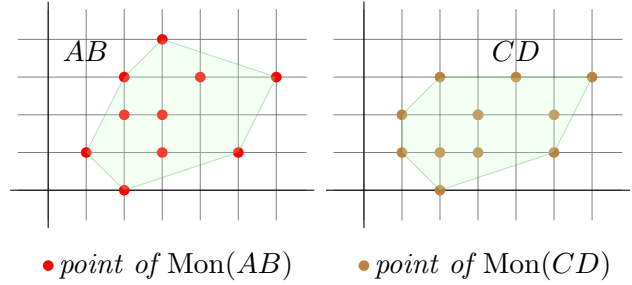
\begin{center}
\begin{tikzpicture}
  \draw[step=.5cm,gray,very thin] (-0.4,-0.4) grid (3.4,2.4);
  \draw (-0.5,0) -- (3.5,0);
  \draw (0,-0.5) -- (0,2.5);
  \fill [red] (0.5cm,0.5cm) circle (2pt);
  \fill [red] (1cm,1cm) circle (2pt);
  \fill [red] (1.5cm,0.5cm) circle (2pt);
  \fill [red] (1cm,0cm) circle (2pt);
  \fill [red] (2.5cm,0.5cm) circle (2pt);
  \fill [red] (1cm,1.5cm) circle (2pt) node[black, above left=2pt] {$AB$};
  \fill [red] (1.5cm,2cm) circle (2pt);
  \fill [red] (2cm,1.5cm) circle (2pt);
  \fill [red] (1.5cm,1cm) circle (2pt);
  \fill [red] (3cm,1.5cm) circle (2pt);
  \filldraw[fill=green!20!white, draw=green!50!black, nearly transparent]
    (0.5,0.5) -- (1,0) -- (2.5,0.5) -- (3,1.5) -- (1.5,2) -- (1,1.5) -- cycle;
    \fill [red] (0cm,-1cm) circle (2pt) node[black, right] {point of
    $\Mon(AB)$};
\end{tikzpicture}
\begin{tikzpicture}
  \draw[step=.5cm,gray,very thin] (-0.4,-0.4) grid (3.4,2.4);
  \draw (-0.5,0) -- (3.5,0);
  \draw (0,-0.5) -- (0,2.5);
  \fill [brown] (0.5cm,0.5cm) circle (2pt);
  \fill [brown] (0.5cm,1cm) circle (2pt);
  \fill [brown] (1cm,1.5cm) circle (2pt);
  \fill [brown] (1cm,0cm) circle (2pt);
  \fill [brown] (1cm,0.5cm) circle (2pt);
  \fill [brown] (1.5cm,1cm) circle (2pt);
  \fill [brown] (1.5cm,0.5cm) circle (2pt);
  \fill [brown] (1.5cm,1cm) circle (2pt);
  \fill [brown] (2cm,1.5cm) circle (2pt) node[black, above=2pt] {$CD$};
  \fill [brown] (2.5cm,0.5cm) circle (2pt);
  \fill [brown] (2.5cm,1cm) circle (2pt);
  \fill [brown] (3cm,1.5cm) circle (2pt);
  \filldraw[fill=green!20!white, draw=green!50!black, nearly transparent]
    (0.5,0.5) -- (1,0) -- (2.5,0.5) -- (3,1.5) -- (1,1.5) -- (0.5,1)
    -- cycle;
    \fill [brown] (0cm,-1cm) circle (2pt) node[black, right] {point of
    $\Mon(CD)$};
\end{tikzpicture}
\end{center}
Then,
\begin{align*}
  AB+CD = & (XY+X^2+X^2Y^2+X^3Y+X^5Y +
  X^2Y^3+X^3Y^2 +X^3Y^4 \\
  & +X^4Y^3+X^6Y^3)-(XY+X^2+X^3Y+X^5Y+XY^2 \\
  & +X^2Y+2X^3Y^2+X^5Y^2+X^2Y^3+X^4Y^3+X^6Y^3) \\
  = & X^2Y^2+X^3Y^4-XY^2-X^2Y-X^3Y^2-X^5Y^2 \\
\end{align*}
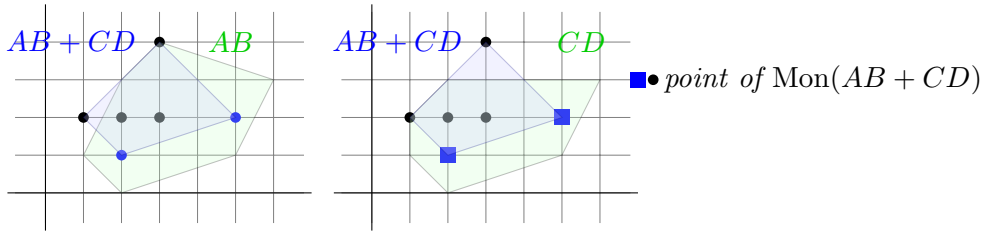
\begin{center}
\begin{tikzpicture}
  \draw[step=.5cm,gray,very thin] (-0.4,-0.4) grid (3.4,2.4);
  \draw (-0.5,0) -- (3.5,0);
  \draw (0,-0.5) -- (0,2.5);
  \fill (1cm,1cm) circle (2pt);
  \fill (1.5cm,2cm) circle (2pt) node[blue, left=5pt] {$AB+CD$} 
  node[green!80!black, right=0.5cm] {$AB$};
  \fill (0.5cm,1cm) circle (2pt);
  \fill [blue] (1cm,0.5cm) circle (2pt);
  \fill (1.5cm,1cm) circle (2pt);
  \fill [blue] (2.5cm,1cm) circle (2pt);
  \filldraw[fill=green!20!white, draw=black, nearly transparent]
    (0.5,0.5) -- (1,0) -- (2.5,0.5) -- (3,1.5) -- (1.5,2) -- (1,1.5) -- cycle;
  \filldraw[fill=blue!20!white, draw=blue!50!black, nearly transparent]
    (0.5,1) -- (1,0.5) -- (2.5,1) -- (1.5,2) -- cycle;
\end{tikzpicture}
\begin{tikzpicture}
  \draw[step=.5cm,gray,very thin] (-0.4,-0.4) grid (3.4,2.4);
  \draw (-0.5,0) -- (3.5,0);
  \draw (0,-0.5) -- (0,2.5);
  \fill (1cm,1cm) circle (2pt);
  \fill (1.5cm,2cm) circle (2pt) node[blue, left=5pt] {$AB+CD$}
  node[green!80!black, right=0.8cm] {$CD$};
  \fill (0.5cm,1cm) circle (2pt);
  \fill [blue] (0.9cm,0.4cm) rectangle (1.1cm,0.6cm);
  \fill (1.5cm,1cm) circle (2pt);
  \fill [blue] (2.4cm,0.9cm) rectangle (2.6cm,1.1cm);
  \filldraw[fill=green!20!white, draw=black, nearly transparent]
    (0.5,0.5) -- (1,0) -- (2.5,0.5) -- (3,1.5) -- (1,1.5) -- (0.5,1) -- cycle;
  \filldraw[fill=blue!20!white, draw=blue!50!black, nearly transparent]
    (0.5,1) -- (1,0.5) -- (2.5,1) -- (1.5,2) -- cycle;
    \fill [blue] (3.4cm,1.4cm) rectangle (3.6cm,1.6cm);
    \fill [black] (3.7cm,1.5cm) circle (2pt) node[black, right] {point of
    $\Mon(AB+CD)$};
\end{tikzpicture}
\end{center}
The two rectangle points lie on the convex hull of $\Mon(AB+CD)$, but do not lie
on the convex hulls of $\Mon(AB)$ or $\Mon(CD)$. 
\end{example}
\begin{example}
We define $f(X,Y)=1+X^2Y+Y^2X$, $g(X,Y)=1+X^4Y+XY^4$ and we consider $\Mon(fg-1)$.
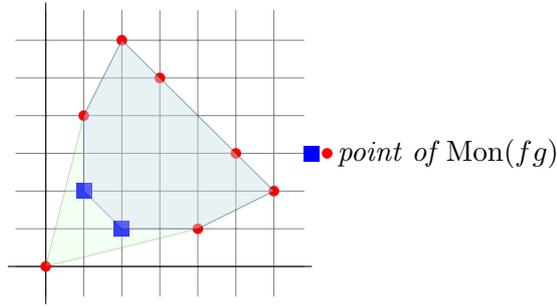
\begin{center}
\begin{tikzpicture}
  \draw[step=.5cm,gray,very thin] (-0.4,-0.4) grid (3.4,3.4);
  \draw (-0.5,0) -- (3.5,0);
  \draw (0,-0.5) -- (0,3.5);
  \fill [red] (0cm,0cm) circle (2pt);
  \fill [blue] (0.9cm,0.4cm) rectangle (1.1cm,0.6cm);
  \fill [blue] (0.4cm,0.9cm) rectangle (0.6cm,1.1cm);
  \fill [red] (2cm,0.5cm) circle (2pt);
  \fill [red] (3cm,1cm) circle (2pt);
  \fill [red] (2.5cm,1.5cm) circle (2pt);
  \fill [red] (0.5cm,2cm) circle (2pt);
  \fill [red] (1cm,3cm) circle (2pt);
  \fill [red] (1.5cm,2.5cm) circle (2pt);
  \filldraw[fill=green!20!white, draw=green!50!black, nearly transparent]
    (0,0) -- (2,0.5) -- (3,1) -- (1,3) -- (0.5,2) -- cycle;
  \filldraw[fill=blue!20!white, draw=blue!50!black, nearly transparent]
    (0.5,1) -- (1,0.5) -- (2,0.5) -- (3,1) -- (1,3) -- (0.5,2) -- cycle;
    \fill [blue] (3.4cm,1.4cm) rectangle (3.6cm,1.6cm);
    \fill [red] (3.7cm,1.5cm) circle (2pt) node[black, right] {point of
    $\Mon(fg)$};
\end{tikzpicture}
\end{center}
The two rectangle points lie on the convex hull of $\Mon(fg-1)$, but do not lie
on the convex hull of $\Mon(fg)$. 

\end{example}

Conjecture~\ref{tauconj} implies that the permanent is hard for arithmetic circuits. In fact, a significantly weaker bound on the number of edges would be sufficient:
\begin{theorem} \label{permanent}
Assume that for some universal constant $c<2$,
 the upper bound $2^{(m+ \log kt)^c}$ on the number of edges
of $\Newt(f)$ holds true for polynomials of the form~(\ref{sps})
whenever the product $kmt$ is sufficiently large.
Then the permanent is not computable by polynomial size arithmetic
circuits.
\end{theorem}
For instance, an upper bound of the form $2^{O(m)} (kt)^{O(1)}$ 
would be sufficient.
 Note that the parameter $m$ plays a very different role 
than the parameters $k$ and $t$.

\begin{remark}
  Theorem~\ref{permanent} holds true even in the case of a field of
  positive characteristic different from two. In characteristic two, the
  permanent is equal to the determinant and is therefore easy for
  arithmetic circuits. However, the conclusion of Theorem~\ref{permanent} 
remains true in characteristic two if we replace the permanent by 
the Hamiltonian polynomial (or any other polynomial family which is $\vnp$-complete in characteristic two).
\end{remark}

\section{Proof of Theorem~\ref{permanent},  and a Refinement} \label{proof}

Consider the polynomial 
\begin{equation} \label{hardpoly}
f_n(X,Y)=\sum_{i=0}^{2^n-1} X^iY^{i^2}.
\end{equation}
 The Newton polygon 
of $f_n$ has exactly $2^{n}$ edges.

Our proof of Theorem~\ref{permanent} is by contradiction. Assuming that
the permanent is computable by 
polynomial size arithmetic circuits
we will show that $f_n$ can be put under form~(\ref{sps}) with 
$k=n^{O(\sqrt{n} \log n)}$, $t=n^{O(\sqrt{n} \log n)}$ and $m=O(\sqrt{n})$.
Note that the upper bound on $m$ is much smaller than those on $k$ and $t$.
Then, from the assumption in Theorem~\ref{permanent} we conclude that
$\Newt(f_n)$ has at most $2^{(m+ \log kt)^c}$ edges. 
This is a contradiction since for large enough~$n$, this upper bound is 
smaller than the actual number of edges of $\Newt(f_n)$, namely, $2^n$
(here, we use the fact that the constant $c$ in Theorem~\ref{permanent} 
is smaller than 2).

Reduction of arithmetic circuits to depth 4 is an important ingredient 
in the proof of Theorem~\ref{permanent}. This phenomenon was discovered 
by Agrawal and Vinay~\cite{AgraVinay08}.
We will use it under the following form~\cite{Koi12} (recall that a 
depth 4 circuit is a sum of products of sums of products of inputs;
sum and product gates may have arbitrary fan-in).
\begin{theorem} \label{circuit2depth4}
Let $C$ be an arithmetic circuit of size $t$ computing a polynomial of 
degree $d$.
There is an equivalent depth four circuit $\Gamma$ of size
$t^{O(\sqrt{d} \log d)}$ with multiplication gates of fan-in $O(\sqrt{d})$.
\end{theorem}
Note that Theorem~3 of~\cite{Koi12} provides this bound for the case
where $d$ is the so-called ``formal degree'' of $C$ rather than 
the degree of the polynomial computed by $C$. 
Theorem~\ref{circuit2depth4} as stated above can then be derived by an 
application of the standard homogenization trick (see e.g. Proposition~5 and
Theorem~5 in~\cite{Koi12}).
It was recently shown~\cite{Tavenas}
 that the size bound for $\Gamma$ can be reduced from
$t^{O(\sqrt{d} \log d)}$ to $t^{O(\sqrt{d})}$ when $d$ is polynomially
bounded in $t$; this improvement preserves
the $O(\sqrt{d})$ bound on the fan-in of multiplication gates.

We can now complete the proof of Theorem~\ref{permanent}. The idea of
this proof is based on the one of Proposition 2 in~\cite{Koi10a}, 
but is simpler since it doesn't use the counting hierarchy.
A similar simplification can be done for the real 
$\tau$-conjecture~\cite{TavenasThese}.

We can expand the exponents $i$ and $i^2$ of~(\ref{hardpoly}) in base
2. 
This leads to the equality
\begin{equation} \label{plugin}
f_n(X,Y)=h_n(X^{2^0},X^{2^1},\ldots,X^{2^{n-1}},Y^{2^0},\ldots,Y^{2^{2n-1}})
\end{equation}
where $h_n(x_0,x_1,\ldots,x_{n-1},y_0,\ldots,y_{2n-1})$
 is the multilinear polynomial
$$\sum_{\alpha\in\{0,1\}^{n},\beta\in\{0,1\}^{2n}} a(n,\alpha,\beta)
 x_0^{\alpha_0}x_1^{\alpha_1} 
\cdots x_{n-1}^{\alpha_{n-1}}y_0^{\beta_0}\cdots
y_{2n-1}^{\beta_{2n-1}}.$$ 
Here the exponents $\alpha_j$, $\beta_j$ denote the coordinates
of the vectors $\alpha$, $\beta$.
The coefficient $a(n,\alpha,\beta)$ equals $1$ if the vector $\beta$ corresponds
to the square of the vector $\alpha$
and $0$ otherwise.
Note that $h_n$ is a polynomial in $3n$ variables.
The coefficients $a(n,\alpha,\beta)$ can be computed in time polynomial in $n$.
By Valiant's criterion~\cite{Burg}, this implies that the polynomial family
$(h_n)$ belongs to the complexity class~$\vnp$.
Since the permanent is $\vnp$-complete and is assumed to have polynomial-size
circuits, $(h_n)$ also has polynomial-size circuits.
By Theorem~\ref{circuit2depth4}, it follows that the polynomials $h_n$
are computable by depth 4 circuits of size $n^{O(\sqrt{n} \log n)}$
with multiplication gates of fan-in $O(\sqrt{n})$.
Using~(\ref{plugin}), we can plug in powers of $X$, $Y$ and powers of 2
into those circuits to express $f_n$ as a sum of products like in~(\ref{sps}).
The resulting parameters $k$ and $t$ are of order $n^{O(\sqrt{n} \log n)}$, and
$m=O(\sqrt{n})$.
As explained at the beginning of this section, this leads to a contradiction
with the assumption in Theorem~\ref{permanent}.~$\Box$

In the remainder of this section we give a refinement of 
Theorem~\ref{permanent}. We show that it suffices to bound the number 
of edges of the Newton polygons of sums of {\em powers} of sparse polynomials
in order to obtain a lower bound for the permanent. 
\begin{theorem} \label{powers}
Fix a universal constant $c<2$, and assume that  we have the upper bound 
$2^{(m+ \log kt)^c}$ on the number of edges
of $\Newt(f)$ for polynomials of the form
\begin{equation} 
f(X,Y)=\sum_{i=1}^k a_i{f_i}(X,Y)^m
\end{equation}
where $a_i \in \kk$ and the $f_i$ have at most $t$ monomials 
(as in Theorem~\ref{permanent},
we require this upper bound to hold only if $kmt$ is sufficiently large).
Then the permanent is not computable by polynomial size arithmetic circuits.
\end{theorem}
Clearly, we can assume that all the coefficients $a_i$ are equal to 1
(multiply $f_i$ by a $m$-th root of $a_i$ 
if necessary).

Theorem~\ref{powers} is an easy consequence of Theorem~\ref{permanent} 
and Fischer's formula~\cite{fischer94}. 
This formula shows that any monomial $z_1z_2\cdots z_m$
can be expressed as a linear combination of $2^{m-1}$ powers of linear 
forms.
\begin{lemma} \label{fischer}
For any $m$, we have 
$$2^{m-1}m!z_1z_2\cdots z_m=\sum_{r=(r_2,\ldots,r_m) \in \{-1,1\}^{m-1}}
(\prod_{i=2}^m r_i)(z_1+\sum_{i=2}^m r_i z_i)^m.$$
\end{lemma}
Note that the exponential blowup entailed by Fischer's formula is acceptable
because we will apply it with a value of $m$ which is small compared 
to the main complexity parameter $n$, i.e., with $m=O(\sqrt{n})$.
The idea of using Fischer's formula to turn a product into sums of
powers comes from~\cite{GuptaKKS13,Kayal12}. 
\begin{proof}[Proof of Theorem~\ref{powers}]
We show that the assumption in Theorem~\ref{powers} implies that of
Theorem~\ref{permanent}. Consider therefore a polynomial $f$ of the
form~(\ref{sps}).
We rewrite it as a sum of powers by applying Lemma~\ref{fischer} to each 
of the $k$ products in~(\ref{sps}).
This yields an identity of the form
$$f(X,Y)=\sum_{i=1}^{k'} a_i{f_i}(X,Y)^m$$
where $a_i \in \kk$, the $f_i$ have at most $mt$ monomials, 
and $k'=2^{m-1}k$.
We are now in position to apply Theorem~\ref{permanent}: $\Newt(f)$ 
has at most $2^{(m+ \log k't')^c}$ edges. 
For any constant $c'>c$, this is less than $2^{(m+ \log kt)^{c'}}$ if
$kmt$ is sufficiently large. We have therefore derived the hypothesis
of Theorem~\ref{permanent} from that of Theorem~\ref{powers}, and
we can conclude that the permanent is hard for arithmetic circuits.
\end{proof}

\begin{remark}
  For Theorem~\ref{powers}, the fact that $\kk$ has characteristic $0$ is
  important. Indeed,  in positive characteristic Lemma~\ref{fischer} 
  does not allow to rewrite a monomial $z_1z_2 \cdots z_m$ 
as a linear combination of powers of linear forms.
\end{remark}

\begin{remark}
As pointed out in the introduction, we gave in~\cite{Koi10a} similar results
for real roots of univariate polynomials rather than for Newton polygons
of bivariate polynomials. More precisely, let us measure the size of 
a sum of products of sparse polynomials by $s=kmt$.
This definition of ``size''
applies to bivariate polynomials of the form~(\ref{sps}) as well 
as to their univariate analogues. We showed 
 that for any constant $c<2$, a $2^{(\log s)^c}$ upper bound 
on the number of real roots implies that the permanent is hard for arithmetic
circuit (see Conjecture~3 in~\cite{Koi10a}  and the remarks following it). 
In fact, the same proof shows than 
an upper bound of the form $2^{(m+\log kt)^c}$ as in Theorem~\ref{permanent}
is sufficient. This is clearly a better way of stating our result since
it allows for a much worse dependency of the number of real roots with
respect to $m$.
Moreover, as in Theorem~\ref{powers} 
it is sufficient to establish this bound 
for sums of powers.
As in the proof of Theorem~\ref{powers}, 
this follows from a straightforward application of Fischer's formula. 
\end{remark}

\section{Upper Bounds from Convexity Arguments} \label{bounds}

In this section we improve the coarse upper bound $k.t^m$ upper bound on 
the number of edges of Newton polygons of polynomials of the form~(\ref{sps}).
These results apply to fields of arbitrary characteristic.
Our main tool is a result of convex geometry~\cite{EPRS08}.
\begin{theorem} \label{convex}
Let $P$ and $Q$ be two planar point sets with $|P|=s$ and $|Q|=t$.
Let $S$ be a subset of the Minkowski sum $P+Q$. If $S$ is convexly independent (i.e., its elements are the vertices of a convex polygon) we have
$|S|=O(s^{2/3}t^{2/3}+s+t)$.
\end{theorem}
It is known that this upper bound is optimal up to constant 
factors~\cite{BBFKOTT10} 
(a non-optimal lower bound was also given in~\cite{SV10}).
A linear upper bound is known~\cite{HOR07} for the case where $P$ and $Q$ are convexly
independent.

We first consider sums of products of two polynomials.
\begin{theorem} \label{2polys}
Consider   a bivariate polynomial
 $f  \in \kk[X,Y]$ 
 of the form
\begin{equation}
f(X,Y)=\sum_{i=1}^k f_i g_i(X,Y)
\end{equation}
where the $f_i$ have at most $r$ monomials and the 
 $g_i$ have at most $s$ monomials.
The Newton polygon  of  $f$ has $O(k(r^{2/3}s^{2/3}+r+s))$ edges.
\end{theorem}
\begin{proof}
Let $S_i$ be the set of points in the plane corresponding to the monomials of $f_i g_i$ which appear in $f$ with a nonzero 
coefficient. Since $\Newt(f)$ is the convex hull
of $\bigcup_{i=1}^k \conv(S_i)$, it is enough to bound the number of 
vertices of $\conv(S_i)$.
Those vertices form a convexly independent subset 
of the Minkowski sum $\Mon(f_i)+\Mon(g_i)$. By Theorem~\ref{convex}, it follows
that $\conv(S_i)$ has $O(r^{2/3}s^{2/3}+r+s)$ vertices.
Multiplying this estimate by $k$ yields an upper bound on the number of
vertices and edges of $\Newt(f)$.
\end{proof}
From this result it is easy to derive an upper bound for the general
case, where we have products of $m \geq 2$ polynomials. 
We just divide the $m$ factors into two groups of approximately $m/2$ factors,
and in each group we expand the product by brute force.
\begin{theorem} \label{mpolys}
Consider   any bivariate polynomial
 $f  \in \kk[X,Y]$ 
 of the form
\begin{equation} 
f(X,Y)=\sum_{i=1}^k \prod_{j=1}^m f_{ij}(X,Y)
\end{equation}
where $m \geq 2$ and the $f_{ij}$ have at most $t$ monomials.
The Newton polygon  of  $f$ has $O(k.t^{2m/3})$ edges.
\end{theorem}
\begin{proof}
As suggested above, we write each of the $k$ products 
as a product of two polynomials $F_i=\prod_{i=1}^{\lfloor m/2 \rfloor} f_i$ and
$G_i=\prod_{i=1}^{\lceil m/2 \rceil} f_i$. 
We can now apply Theorem~\ref{2polys} to $f=\sum_{i=1}^k F_i G_i$,
with $r=t^{\lfloor m/2 \rfloor}$ and $s=t^{\lceil m/2 \rceil}$.
In the resulting $O(k(r^{2/3}s^{2/3}+r+s))$ upper bound the term 
$kr^{2/3}s^{2/3}$ dominates since $r^{2/3}s^{2/3}=
t^{2(\lfloor m/2 \rfloor+\lceil m/2 \rceil)/3}=t^{2m/3}$ and $m \geq 2$.
\end{proof}
In order to avoid the brute force expansion in the proof of this theorem
it is natural to consider for each $i$ a convexly independent subset $S_i$ 
of the Minkowski sum of the $m$ sets $\Mon(f_{i1}),\ldots,\Mon(f_{im})$.
This is exactly the open problem at the end of~\cite{BBFKOTT10}:
determine the maximal cardinality $M_m(t)$ of a convexly independent subset
of the Minkowski sum of $m$ sets $P_0,\ldots,P_{m-1}$ 
of $t$ points in the Euclidean plane.
For instance, the lower bound of~\cite{BBFKOTT10} 
combined with the upper bound of~\cite{EPRS08} shows that 
$M_2(t)=\Theta(t^{4/3})$.
Unfortunately, we shall see that $M_m(t)=t^{\Omega(m)}$, so that brute
force expansion is not very far from the optimum.

\begin{example} \label{base} Fix an integer $b \geq 2$. Let $P_k$ be
  the $b^2 \times b$ grid made of the integer points $(x,y)$ such that:
\begin{itemize}
\item[-] all the digits in base $b^2$ of $x$ are equal to zero, except possibly
the 
digit of weight $b^{2k}$;

\item[-]  all the digits in base $b$ of $y$ are equal to zero, except possibly the 
digit of weight $b^k$.
\end{itemize}
More explicitly,
  \begin{align*}  
    P_k = \{(b^{2k}.i,b^k.j);\ 0 \leq i \leq b^2 -1 \textrm{ and } 0
    \leq j \leq b-1\}.
  \end{align*}
  Clearly, the Minkowksi sum $P_0+\ldots+P_{m-1}$ is the grid
  $\{0,\ldots,b^{2m}-1\} \times \{0,\ldots, b^m-1\}$ of size 
  $b^{2m}\times b^m$.
\end{example}

The next lemma (which is certainly not optimal) shows how to find a 
fairly large set of convexly independent points in a grid.
\begin{lemma} \label{grid}
If $n(n-1)/2 < M$ and $n<N$ it is possible to find $n$ convexly independent 
points in the grid $M \times N$.
\end{lemma}
\begin{proof}
We start from the origin and build a sequence of $n-1$ line segments.
The $i$-th segment has horizontal length $i$ and slope $1/i$.
We can keep going as long as we do not go out of the grid, i.e., as long
as $n(n-1)/2 < M$ and $n<N$. Altogether, the $n-1$ segments have $n$ endpoints and 
they are convexly independent.
\end{proof}

\begin{proposition} \label{lb} For all $m$ and infinitely many values
  of $t$ we have:
  $$M_m(t) \geq t^{m/3}-1.$$
\end{proposition}
\begin{proof}
  From Example~\ref{base} and Lemma~\ref{grid} (choosing $M=b^{2m}$
  and $N=b^m$) 
  we have $M_m(b^3) \geq n$ if $n(n-1)<2b^{2m}$ and $n<b^m$. Hence $M_m(b^3) \geq b^m-1$.
  The result follows by setting $t=b^3$.
\end{proof}
This result shows that other ingredients than Theorem~\ref{convex} will be needed to
answer Conjecture~\ref{tauconj} positively.
A similar argument can be made for the case 
where the sets $P_0,P_1,\ldots,P_{m-1}$ in the Minkowski sum are all equal
(this is a natural case to look at in light of Theorem~\ref{powers}, which
shows that it suffices to deal with sums of powers in order to obtain
a lower bound for the permanent).
More precisely, let $M'_m(t)$ be the maximal cardinality of a convexly
independent subset of an $m$-fold Minkoski sum $P+P+\cdots+P$ where
$P$ is a set of at most $m$ points. By definition we have 
$M'_m(t) \leq M_m(t)$. In the other direction we have 
$M'_m(t) \geq M_m(\lfloor t/m \rfloor)$:
just replace the $m$ sets of size
$\lfloor t/m \rfloor$ by their union.
Hence we have $M'_m(t) \geq \lfloor t/m\rfloor^{m/3} -
1$.

\section{Final Remarks}

In this paper we have proposed a conjecture on the number of edges of the 
Newton polygon of a sum of products of sparse polynomials; 
and we have shown in Theorem~\ref{permanent} that even a weak version 
of this conjecture implies a lower bound for the permanent.
We conclude with a couple of additional open problems.
\begin{enumerate}
\item Consider two polynomials $f, g \in \kk[X,Y]$ with at most $t$
  monomials each. What is the maximum number of edges on the Newton
  polygon of $fg+1$? The difficult case appears when the constant term
  of $fg$ equals $-1$. Theorem~\ref{2polys} provides a $O(t^{4/3})$
  upper bound, but as far as we know the ``true'' bound could be
  linear in $t$. In the appendix we prove a linear upper bound under
  the assumption that $f$ and $g$ have the same supports (i.e.,
  $\Mon(f)=\Mon(g)$)
and that the square of any nonconstant
monomial appearing in $f$ and $g$ does not appear in $fg$.

\item More generally, what is the 
maximum number of edges on the Newton polygon of $f_1\ldots f_m+1$,
where the $f_i$ again have at most $t$ monomials?
Theorem~\ref{mpolys} provides a $O(t^{2m/3})$ upper bound, 
but the true bound could be of the form $2^{O(m)} t^{O(1)}$; it could 
even be polynomial in $m$ and $t$, as implied by Conjecture~\ref{tauconj}.
\end{enumerate}

{\small

\section*{Acknowledgments}

Proposition~\ref{lb} arose  from a discussion with Mark Braverman and
an improvement was suggested by an anonymous referee.
We thank the three referees for suggesting several improvements in
the presentation of the paper.



\appendix
\newpage
\section*{Appendix: the Newton polygon of $fg+1$}

In this section we denote by $0$ the point in the plane with
coordinates $(0,0)$.

We give here (in Theorem~\ref{sameth}) a linear upper bound assuming the
following two properties:
\begin{itemize}
\item[(i)]
The polynomials $f$ and $g$ have the same support, i.e., $\Mon(f)=\Mon(g)$.
We denote by $\{p_0,\ldots,p_{t-1}\}$ this common support.

\item[(ii)] If $f$ and $g$ have a constant term we assume without loss 
of generality that $p_0=0$ and we add the following requirement:
if $p_j$ is an extremal point of $\conv(p_1,p_2,\ldots,p_{t-1})$ 
then $2p_j$ is not in the support of $f$ and $g$.
\end{itemize}
We do not know how to prove a linear upper bound assuming only (i).
Condition~(ii) is satisfied in particular when the points in $\Mon(f)=\Mon(g)$
are convexly independent.

The interesting case, which we consider first,  is when $f$ and $g$ have a constant term 
but $fg+1$ has no constant term.
As explained above we assume that $p_0$ corresponds to the constant
terms of $f$ and $g$, i.e., $p_0=0$.
Under these hypotheses we have the following result.
\begin{proposition} \label{sameprop}
Under assumptions (i) and (ii), $$\Newt(fg+1)=\conv(2p_1,\ldots,2p_{t-1},(p_i)_{i \in I})$$
where $(p_i)_{i \in I}$ is the subset of those monomials in $\Mon(f)$
which appear in $fg+1$ with a nonzero coefficient.
\end{proposition}
\begin{proof}
We first prove the inclusion from left to right. Since $fg+1$ has no constant term, all monomials 
of $fg+1$ are of the form $p_i+p_j$ where $i \geq 1$ or $j \geq 1$.
Consider first the case where $i$ and $j$ are both nonzero.
If $i=j$ this monomial appears in the right-hand side, and
if $i \neq j$ it is the middle point of two points (namely, $2p_i$ and $2p_j$) 
appearing in the right-hand side.
The remaining case is when $i=0$ or $j=0$.
If e.g. $j=0$ we have $p_i+p_j=p_i$ and we see from the definition of $I$ 
that this monomial also appears
in the right-hand side.

Now we prove the inclusion from right to left.
Again by definition of $I$, all the $p_i$ with $i \in I$ are monomials
of $fg+1$. Hence it remains to show that 
$$\conv(2p_1,\ldots,2p_{t-1}) \subseteq \Newt(fg+1).$$
The left-hand side can be written as $\conv((2p_j)_{j \in J})$ where
the $(p_j)_{j\in J}$ form a convexly independent subset of 
$\{p_1,\ldots,p_{t-1}\}$.
Any monomial of the form $2p_j$ with $j \in J$ appears in $fg+1$ 
with a nonzero coefficient because it can be obtained in a unique way
by expansion of the product $fg$.
Assume indeed that $2p_j=p_i+p_k$ with $i \neq k$.
Then $p_j$ is the middle point of $p_i$ and $p_k$.
If $i \geq 1$ and $k \geq 1$,  this is impossible
by construction of $J$. 
If $i=0$ or $k=0$, this is also impossible by hypothesis (ii).
We thus have $\conv((2p_j)_{j \in J}) \subseteq \Newt(fg+1)$, and the proof is complete.
\end{proof}
We note that this proposition does not hold without assumption (ii), as shown by the following example:
take $f=1+X^2Y+XY^2+(1/2)X^2Y^4+(1/2)X^4Y^2$ and
$g=-1+X^2Y+XY^2-(1/2)X^2Y^4-(1/2)X^4Y^2$.
Then $fg+1=2X^3Y^3-(1/2)X^6Y^6-(1/4)X^4Y^8-(1/4)X^8Y^4$.
The monomial $X^3Y^3$ is a vertex of $\Newt(fg+1)$ but is not of the form $p_i$ or $2p_j$ prescribed by Proposition~\ref{sameprop}.

\begin{theorem} \label{sameth}
Under the same assumptions (i) and (ii) as above,
 $\Newt(fg+1)$ has at most $t+1$ edges where $t$ denotes the number of monomials of $f$ and $g$.
\end{theorem}
\begin{proof}
We continue to denote the common support of $f$ and $g$ by $\{p_0,\ldots,p_{t-1}\}$.
If $0$ does not belong to this support then $\Newt(fg+1)$ is the
convex hull of $\{0\}$ and  $\Newt(fg)$. 
Moreover, $\Newt(fg)=\Newt(f)+\Newt(g)=\conv(2p_0,\ldots,2p_{t-1})$.

If $0$ is in the support and $fg+1$ has a constant term then 
$\Newt(fg+1)=\Newt(fg)$ has at most $t$ edges ($t$ and not $2t$ since $f$ and $g$ have
the same support).

In the remaining case ($0$ is in the support but $fg+1$ has no constant term)
we need to use hypothesis (ii). This case 
is treated in Proposition~\ref{sameprop}.
At first sight it seems that $\Newt(fg+1)$ can have up to $2(t-1)$ vertices,
but the list of possible vertices can be shortened by picking 
a convexly independent subsequence.
More precisely, write $\conv(2p_1,\ldots,2p_{t-1},(p_i)_{i \in I})=\conv((2p_j)_{j \in J},(p_k)_{k \in K})$ where $J \subseteq \{1,\ldots,t-1\}$ and $K \subseteq I$ are chosen so that the points in this
sequence are convexly independent. By the lemma below,
$| J \cap K| \leq 2$. As a result, 
the number of points in the sequence is 
 $|J|+|K| = |J \cup K| + | J \cap K| \leq (t-1)+2=t+1.$
\end{proof}

\begin{lemma}
If $p$, $q$, $r$ are 3 distinct nonzero points in the plane then the 6 points
$p$, $q$, $r$, $2p$, $2q$, $2r$ are not convexly independent.
\end{lemma}
This is clear from a picture and can be proved for instance by 
considering the 4 points $0$, $p$, $q$, $r$.
There are two cases.
\begin{enumerate}
\item If these 4 points are convexly independent, 
assume for instance that $pq$ is a  diagonal
of the quadrangle $0prq$. 
Then the line $pq$ separates $0$ from $r$. As a result, $r \in \conv(p,q,2r)$.

\item If the 4 points are not convexly independent, assume for instance that
$r \in \conv(0,p,q)$. In this case, $2r \in \conv(2p,2q,r)$.~$\Box$
\end{enumerate}

\end{document}